\documentclass[a4paper]{article}
\usepackage[pdftex]{graphicx}
\usepackage{geometry}
\usepackage{amssymb}
\usepackage{amstext}
\usepackage{amsmath}
\usepackage{amsthm}
\usepackage{multicol}
\usepackage{mathrsfs}
\usepackage{mathtools}
\usepackage{nccmath}
\usepackage{tikz}
\usepackage{here}
\usepackage{caption}
\usepackage{setspace}
\usepackage{authblk}
\usepackage{fancybox}
\usepackage{url}

\usetikzlibrary{positioning}
\usetikzlibrary{arrows.meta}

\theoremstyle{plain}
\newtheorem{df}{Definition}
\newtheorem{thm}{Theorem}
\newtheorem{exa}{Example}

\newtheorem{lemma}{Lemma}

\begin{document}
\title{Logic of Awareness for Nested Knowledge}
\author{Yudai Kubono}  
\affil{\textit{Graduate School of Science and Technology, Shizuoka University,} \par \textit{Ohya, Shizuoka 422-8529, Japan} \par \textit{yudai.kubono@gmail.com}}
\date{February 1, 2024}       
\maketitle
\doublespacing

\begin{abstract}
  Reasoning abilities of human beings are limited. Logics that treat logical inference for human knowledge should reflect these limited abilities. Logic of awareness is one of those logics. In the logic, what an agent with a limited reasoning ability actually knows at a given moment (\textit{explicit knowledge}) is distinguished from the ideal knowledge that an agent obtains by performing all possible inferences with what she already knows (\textit{implicit knowledge}) \cite{Fernandez-Fernandez2021-FERAIL-2}. 
  
  This paper proposes a logic for explicit knowledge. In particular, we focus more on nested explicit knowledge, which means another agent's knowledge that an agent actually knows at a given moment. We develope a new formalization of two ideas and propose Kripke-style semantics. The first idea is the effect on an agent's reasoning ability by a state of an agent's awareness. We incorporate a relation on possible worlds called an \textit{indistinguishable relation} to represent ignorance due to lack of awareness. The second idea is a state of each agent's awareness in the other agent's mind. We incorporate a non-empty finite sequence of agents called \textit{a chain of belief for awareness}. Our logic is called \textit{Awareness Logic with Partitions and Chains} ($\mathcal{ALPC}$). Employing an example, we show how nested explicit knowledge is formalized with our logic. Thereafter, we propose the proof system and prove the completeness. Finally, we discuss directions for extending and applying our logic and conclude. Our logic offers a foundation for a formal representation of human knowledge. We expect that the logic can be applied to computer science and game theory by describing and analyzing strategic behavior in a game and practical agent communication.
\end{abstract}


\section{Introduction}

Epistemic Logic ($\mathrm{EL}$), also known as the logic of knowledge, formalizes logical inference for knowledge. The original interest is to elucidate the properties of knowledge. More recently, researchers in fields such as computer science and economics have become interested in the logic as a tool to formally analyze the effect of communicative protocols on knowledge of systems and strategic situations where a player makes a decision with knowledge of the other's action \cite{fagin1995reasoning}. One of the research directions is formalizing the knowledge of beings with limited reasoning abilities, such as humans and computers, rather than abstract entities. Another direction is formalizing knowledge of another agent's knowledge in multi-agent cases, called nested knowledge. This paper proposes a logic for nested knowledge of human beings. We expect that our logic offers a foundation for theoretical research in computer science and game theory. 

$\mathrm{EL}$ is defined as a modal logic, where a necessity operator in a modal logic is interpreted as a knowledge operator. This means that the knowledge operator inherits logical properties that the necessity operator possesses. Propositions are denoted by $\varphi,\psi$ and a knowledge operator with agent $a$ is denoted by $K_a$, which together with a proposition reads ``$a$ knows $\varphi$.'' We may write $K\varphi$ anonymously. In epistemic logic, as well as modal logic, we employ \textit{Modus Ponens} ($\mathrm{MP}$)\footnote{From $\varphi$ and $\varphi\to\psi$, we conclude $\psi$.}, $\mathrm{RN}$\footnote{$\mathrm{RN}$: If $\varphi$ is valid, then agent $a$ knows $\varphi$.}, and $\mathrm{K}$\footnote{$\mathrm{K}$: $K(\varphi\to\psi) \to (K\varphi \to K\psi)$.} for logical inference. This means that when agent $a$ knows $\varphi$, and $\varphi \to \psi$ is the truth, $\psi$ must be her knowledge. This property does not matter when an agent is interpreted as an abstract entity. However, when the concept of an agent is applied to a human being, it matters. A real human's reasoning ability is limited and does not obtain knowledge through such exhaustive reasoning. The gap between knowledge in $\mathrm{EL}$ and human knowledge is called the problem of \textit{logical omniscience}. 

Logic of awareness is one of the solutions to the problem. Awareness Logic \cite{fagin1988belief}, which is the seminal study in this field, distinguishes between knowledge that an agent can actually use for their inference at a given moment and knowledge that an agent cannot actually use for it. The former is called \textit{explicit knowledge}, which means what an agent with a limited reasoning ability actually knows at a given moment. On the other hand, the latter is called \textit{implicit knowledge}, which means the ideal knowledge that an agent obtains by performing all possible inferences with what she already knows. We can say that ``implicit knowledge is the best an agent can do if she were omniscient'' \cite[p.38]{Fernandez-Fernandez2021-FERAIL-2}. In the logic, even if an agent knows $\varphi$, and $\varphi \to \psi$ is the truth, $\psi$ is not necessarily her knowledge anymore because $\varphi$ or $\varphi \to \psi$ may not be explicit knowledge. She may not actually know these propositions at a given time. Agent $a$'s explicit knowledge is denoted as $E_a$, and the implicit knowledge is denoted as $I_a$. The logic has been developed in various fields such as philosophy, computer science, and economics to analyze an agent's limited reasoning \cite{van2015handbook}.

Before going into detail about logic of awareness, we explain the semantical idea in $\mathrm{EL}$. Its semantics are given by a Kripke model. The model is a tuple consisting of a set of possible worlds $W$, an equivalent relation $R_a$ on $W$ for each agent $a$ called an accessibility relation, and a valuation $V$. The truth of a formula $K_a \varphi$ at a world $w$ is defined by for all possible worlds accessible on $R_a$ from $w$, $\varphi$ is true. A possible world represents a possibility, and possible worlds accessible from $w$ on $R_a$ represent all the possibilities that agent $a$ considers. This means a binary relation can be interpreted as a representation of an agent's reasoning ability. If an agent can access only the actual world, she knows everything that holds in actuality. As a human being, an agent accesses possibilities other than the actual one. If an agent accesses a world where $\varphi$ is false, $\neg K_a \varphi$ holds at $w$, which means ``agent $a$ does not know $\varphi$.''


In Awareness Logic \cite{fagin1988belief}, an \textit{awareness set} is incorporated to a Kripke model. This set represents a state of an agent's awareness for formalizing explicit and implicit knowledge. In the logic, implicit knowledge is defined in the same way as the knowledge in $\mathrm{EL}$, and explicit knowledge is defined by implicit knowledge contained in an awareness set, which means implicit knowledge of which an agent is aware. ``$\varphi$ is explicit knowledge'' equals the conjunction of ``$\varphi$ is implicit knowledge'' and ``$\varphi$ is an aware proposition.'' This idea is simple and intuitive and has become one of the main approaches in logic of awareness. 

Nested knowledge, which is knowledge of another agent's knowledge, appears in multi-agent cases. Among such knowledge, common knowledge has a distinctive property. This knowledge is defined by infinite iterations of ``every agent knows.'' It is because that ``everyone knows' is not enough, but a mutual acknowledgment of ``everyone knows' is necessary \cite{fagin1995reasoning}. It is known that common knowledge plays a significant role in people's agreement in economic \cite{aumann1976agreeing}. 

We formalize nested explicit knowledge in the field of logic of awareness. We argue that there are two ideas to formalize to represent such knowledge, and we propose a new formalization incorporating these ideas. Note that the scope of this paper is limited to the formalization of nested explicit knowledge, which means it does not include that of common knowledge. 

The first one is the effect of an agent's awareness on her own reasoning. An agent's awareness affects not only whether she is aware of propositions but also her reasoning ability, which means a set of possible worlds that she can access. The previous study captured only the former. When agent $a$ is not aware of $q$ and $\neg q$, she cannot consider two possibilities: the worlds with only a different valuation for $q$. It is because she should not identify the distinction in the first place. The states of awareness of $a$ and $b$ may differ, in which case $a$'s reasoning ability $R_a$ limited by the state of $a$'s own awareness may differ from $R_a$ reflecting the state of $b$'s awareness. In the model depicted in Figure 1, $a$ herself distinguishes between $w_1$ and $w_2$, but an imaginary $a$ with the same awareness as $b$ does not distinguish them. Thus, the notion of an accessible relation needs to be extended.

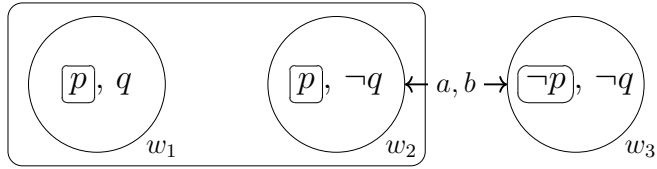
\begin{figure}[h]
  \centering
\scalebox{0.9}[0.9]{
\begin{tikzpicture}
\draw[rounded corners=6pt] (-1.3,1.2)--(4.8,1.2)--(4.8,-1.2)--(-1.3,-1.2)--cycle;
\draw(0,0)circle(1);
\draw(0.95,-0.95)node{{\large $w_1$}};
\draw(3.5,0)circle(1);
\draw(4.45,-0.95)node{{\large $w_2$}};
\draw(7,0)circle(1);
\draw(7.95,-0.95)node{{\large $w_3$}};
\draw(0,0)node{{\Large \ovalbox{{$p$}}, $q$}};
\draw(3.5,0)node{{\Large \ovalbox{{$p$}}, $\neg q$}};
\draw(7,0)node{{\Large \ovalbox{{$\neg p$}}, $\neg q$}};
\draw[thick, <->](4.5,0)--(6,0);
\draw(5.25,0)node[fill=white]{{\large $a,b$}};
\end{tikzpicture}\par}
\caption{A two multi-agent model in which a set of possible worlds is $\{w_1,w_2,w_3\}$. Propositions inside a possible world are true propositions in the world. An arrow represents an accessibility relation. Reflexive arrows are omitted. Agent $a$ is aware of $p$ and $q$, but agent $b$ is only aware of $p$. Rounded rectangles that contain a proposition refer to that $b$ is aware of the proposition. A rounded rectangle that contains worlds refers to $b$ does not distinguish the worlds.}
\end{figure}

The second one is a state of each agent's awareness in the other agent's mind. We consider the meaning of nested explicit knowledge and take $E_a E_b p$ for the example.``Agent $a$ explicitly knows that $b$ explicitly knows $p$'' means ``for all possible worlds accessible for $a$, for all possible worlds accessible for $b$, $p$ is ture.'' Their awareness limits agents' reasoning abilities. In $\mathrm{EL}$, it is assumed that both agents have the same states of awareness. However, there is also a case where $a$ does not have a full grasp of $b$'s state of awareness. For explicit knowledge in such a situation, it is necessary to introduce $b$'s state of awareness in $a$'s mind in this sentence where $a$ is the subject. This paper interprets agent $b$'s state of awareness in $a$'s mind as what $a$ believes $b$ is aware of, and this state of awareness can differ from the actual $b$'s state of awareness. To formalize such a nested explicit knowledge, we should introduce a new notation $E_{(a,b)}$ to represent $b$'s explicit knowledge based on $a$'s belief in $b$'s state of awareness. This knowledge is different from $b$'s explicit knowledge $E_b$.

The question arises whether such a knowledge $E_{(a,b)}$, $b$'s explicit knowledge based on $a$'s belief in $b$'s state of awareness, is a mere $a$'s belief. The answer is in the negative. It is because this information has a property of knowledge in $\mathrm{EL}$. For example, $E_{(a)} E_{(a,b)}\varphi \to E_{(a,b)}\varphi$ holds in our logic. This is called $\mathrm{T}$-axiom in modal logic, which is one of the features of knowledge. Thus, such information should be interpreted as knowledge, not belief.
The following is a concrete example of nested explicit knowledge. We use the example in \cite{yudai2022-1} as a reference.

\begin{exa}
  Let agent $a$ be the owner of the store $A$ and agent $b$ be the owner of the store $B$ considering opening their new store in an area. Product costs have risen due to poor harvests, and a reckless expansion leads to a significant loss. Only $a$ is aware of a new wholesaler that supplies the product at half the current price. She believes that $b$ is also aware of the wholesaler as well as herself and believes that $b$ believes she is not aware of the wholesaler. Agent $b$ is not aware of the wholesaler in actuality.
\end{exa}
Let $p_{a/b}$ denote ``$a/b$ opens a new stores'' and $q$ denote ``there is the wholesaler.'' In this setting, both owners explicitly know their own decisions. Agent $a$ actually knows that if $q$ is true, then $p_a$ and $p_b$ should be true. It is because she is aware of the wholesaler, and if such a wholesaler exists in actuality, they should expand their trade area. However, it is not $b$'s explicit knowledge because $b$ does not have any clue to know about $a$'s decision ($p_a$ or $\neg p_a$) and does not actually know $a$'s decision. Moreover, if $a$ actually knows that $q$ is true, imaginary $b$ with the same awareness as $a$ explicitly knows $a$'s decision. The same idea applies to $a$'s belief in $b$'s belief in $a$'s state of awareness and $b$'s belief in $a$'s state of awareness. Each owner's knowledge can be organized as follows.
\begin{itemize}
  \item[$\bullet$] 
  $E_{(a)} p_{a}$ and $E_{(b)} p_{b}$: $a$ and $b$ explicitly know their own decision,
  \item[$\bullet$] 
  $E_{(a)} q \to E_{(a)} p_b$: if $a$ explicitly knows that $q$ is true,  
  $a$ explicitly knows $b$'s decision,
  \item[$\bullet$] 
  $\neg E_{(a)} q \to \neg E_{(a)} p_b$: if $a$ does not explicitly know 
  that $q$ is true, $a$ does not explicitly know $b$'s decision,
  \item[$\bullet$] 
  $\neg E_{(b)} p_a$: in any case, $b$ does not explicitly know $a$'s decision,
  \item[$\bullet$] 
  $E_{(a)} q\to E_{(a)} E_{(a,b)} p_a$: if $a$ explicitly knows that $q$ is true,  $a$ explicitly knows that $b$ explicitly knows $a$'s decision,
  \item[$\bullet$] 
  $E_{(b)}\neg E_{(b,a)} p_b$: in any case, 
  $b$ explicitly knows that $a$ does not explicitly know $b$'s decision,
  \item[$\bullet$] 
  $\neg E_{(a,b,a)} p_b$: in any case, imaginary $a$ with the same awareness as $b$ does not explicitly know $a$'s decision,
  \item[$\bullet$] 
  $E_{(a)} E_{(a,b)} \neg E_{(a,b,a)} p_b$: in any case, $a$ explicitly knows that $b$ explicitly knows that $a$ explicitly knows $a$'s decision.
\end{itemize}

The previous logic cannot represent some explicit knowledge that is seemingly contradictory but is both true, such as $E_{(a,b)} p_a$ and $\neg E_{(b)} p_a$. Our logic can represent such knowledge, and we show it in a later section.

Such knowledge is complicated but valuable when it comes to describing and analyzing a game in game theory. The concept of awareness has been paid attention to analyzing human behaviors in game theory. There are several formalizations and solution concepts of a game where players have different states of awareness, also known as \textit{game with incomplete awareness/unawareness} \cite{feinberg2005games}. In a two-player game with incomplete awareness, a player $a$ reasons another player $b$'s reasoning under the conditions that $b$ may not be aware of all her possible actions, and $b$ may perceive $a$ is not aware of $b$'s possible actions, and so on. The example provided in \cite{feinberg2005games} shows that another player's awareness that a player actually uses on her reasoning affects selecting a strategy. A game that a player perceives she plays may differ by her state of awareness, and Nash equilibrium may also differ. Another player's choice changes according to another player's awareness that a player actually uses on her reasoning.

Models that can distinguish between \textit{awareness of} and \textit{awareness that} have been proposed in the field of philosophy\cite{fernandez2021awareness,grossi2015syntactic}. The models allow us to represent the relation between the concept of knowledge and awareness more accurately. The former, which is ``awareness of'', expresses the notion of awareness in the sense of being able to refer to the information. The latter is the notion of awareness in the sense of acknowledging that the information is true through reasoning or observation. Although both notions are similar, they have different properties. Even if an agent can refer to an apple on a desk, an agent does not necessarily acknowledge that an apple is actually on the desk. In previously proposed logic, explicit knowledge was defined by combining these two notions. In this paper, we focus on ``awareness of.'' This is because the notion is more relevant to the example, and we do not need to consider the other one.

\paragraph*{Overview}
The paper is structured as follows. In Section 2, we define \textit{Awareness Logic with Partitions and Chains} ($\mathcal{ALPC}$). Its semantics are given in the Kripke-style. The semantics has the standard Kriple model for knowledge, and we add an extended awareness set and another equivalence relation. The relation changes depending on an agent's belief in a state of another agent's awareness. The notion of the relation was proposed in \cite{yudai2022-1}. Besides, we show how our logic works with the example introduced in Section 1. In Section 3, we give the proof system \textbf{ALPC} in Hilbert style. For proving the completeness theorem, we use techniques of logic of the modality for transitive closure \cite{van2007dynamic} and the idea of a generated model \cite{chellas1980modal}. In Section 4, we introduce several related work to our logic. Lastly, we discuss possible future directions for extending and applying our logic and conclude.

\section{Awareness Logic with Partitions and Chains}
This section defines our logic \textit{Awareness Logic with Partitions and Chains} ($\mathcal{ALPC}$) and demonstrates to represent the knowledge of the example in Section 1. The logic is based on Awareness Logic \cite{fagin1988belief} and designed to incorporate the two ideas mentioned in Section 1, which are the effect of an agent's awareness on her own reasoning
and a state of each agent's awareness in the other agent's mind. For the former, we extend the equivalence relation in \cite{yudai2022-1}. As for the latter, we introduce a sequence of agents and extend an awareness set in \cite{fagin1988belief}.

\subsection{Language}
First, we define what represents the second idea, which is a state of each agent's awareness in the other agent's mind. We call it a \textit{chain of belief for awareness}. This chain is defined as a non-empty finite sequence of agents $(i_1,\dots,i_n)$. The length of the chain $(i_1,\dots,i_n)$ is $n$. Moreover, an order $\preceq$ on a set of chains $\Theta$ is defined as a partial order satisfying the following conditions: 
\begin{itemize}
  \item[$\bullet$] $\theta \preceq \theta'$ for 
  $\theta' = \theta\cdot\theta''$ or $\theta = \theta'$; 
  \item[$\bullet$] $\theta \preceq \theta'$ and 
  $\theta' \preceq \theta$ 
  for $\theta = (i_1,\dots,i_k,i_{k+1},\dots i_n)$ 
  such that $i_k = i_{k+1}$, and 
  $\theta' =(i_1,\dots,i_k,\dots i_n)$,
\end{itemize} 
where $\theta,\theta',\theta''\in \Theta$, and $\cdot$ refers to a concatenation of sequences. The condition of $\theta \preceq \theta'$ and $\theta' \preceq \theta$ is denoted as $\theta \simeq \theta'$. An order $\prec$ is defined in the usual way: $\theta\prec\theta'$ \textit{ iff }$\theta\preceq\theta'$ and $\theta'\not\preceq\theta$. 
The second condition means that if the same agent sequentially occurs in a chain $\theta$, it is said to be symmetric to the chain in which one of the agents is removed $\theta'$. For instance, if $\theta=(a,b,c)$ and $\theta'=(a,b)$, then $\theta'\prec \theta$; if $\theta=(a,b,b)$ and $\theta'=(a,b)$, then $\theta \simeq \theta'$, even though $\theta\not=\theta'$.

Next, we define the language. Let $\mathcal{P}$ be a countable set of atomic propositions, $\mathcal{G}$ be a finite set of agents, and $\Theta$ be a finite set of chains of belief for awareness.
The language $\mathcal{L}_{(\mathcal{P,G},\Theta)}$ is the set of formulas generated by the following grammar: 
  \begin{align*}
    \mathcal{L}_{(\mathcal{P,G},\Theta)} \ni\varphi
    ::=  p &\mid \neg\varphi \mid \varphi\wedge\varphi\mid 
     I_i\varphi\mid [\mathop{\approx}]_{\theta}\varphi 
     \mid C_{\theta}\varphi \mid 
     \triangle_{\theta} \varphi,
  \end{align*}
where $p \in \mathcal{P}$, $i\in\mathcal{G}$, $\theta\in \Theta$,  $\triangle\in\{A,K\}$, and for $\triangle_{\theta}$, if the form of $\triangle_{\theta'} \psi$ occurs in $\varphi$, then $\theta\preceq \theta'$. Other logical connectives $\vee$, $\to$, and $\leftrightarrow$ are defined in the usual manner. $A_{\theta}$, $I_i$, and $E_{\theta}$ are called an awareness operator, an implicit knowledge operator, and an explicit knowledge operator respectively. Let $\theta$ be $(i_1,\dots,i_n)$. Notationally, 
\begin{itemize}
  \item[$\bullet$] $A_{(i_1,\dots,i_n)}\varphi$ reads   
  ``$i_1$ believes $\cdots$ $i_{n-1}$ believes that $i_n$ is aware 
  of $\varphi$.''
  \item[$\bullet$] $I_i\varphi$ reads ``$i$ knows $\varphi$ implicitly.''
  \item[$\bullet$] $E_{(i_1,\dots,i_n)}\varphi$ reads ``imaginary $i_n$ with the awareness that $i_1$ believes $\cdots$ $i_{(n-1)}$ believes that $i_n$ has explicitly knows $\varphi$.''    
\end{itemize}
The operators $[\approx]_{\theta}$ and $C_{\theta}$ are special operators. These are used to define explicit knowledge and used in proofs in a later section. The operator $[\approx]_{(i_1,\dots,i_n)}$ can be interpreted as an operator that refers to true information in all the worlds where imaginary $i_n$ with the awareness that $i_1$ believes $\cdots$ $i_{(n-1)}$ believes that $i_n$ has does not distinguish. The operator $C_{(i_1,\dots,i_n)}$ can be interpreted as an operator that refers to a kind of $i_n$'s implicit knowledge. As the difference from $I_i$ operator, $I_i$ refers to $i$'s knowledge that $i$ can obtain in the case where $i$ is aware of all the propositions. On the other hand, $C_{(i_1,\dots,i_n)}$ refers to imaginary $i_n$'s knowledge in the case where she is only aware of what $i_1$ believes $\cdots$ $i_{(n-1)}$ believes that she has. Besides, ``the awareness that $i_1$ believes $\cdots$ $i_{(n-1)}$ believes that $i_n$ has'' is simpliy denoted by ``with the awareness of $\theta$.''

We take up some formulas to understand the language and the meanings of operators. The formula $E_{(i)} p$ is an element of the language and means ``$i$ explicitly knows $p$.'' The formula $E_{(i)} E_{(i,j)} p$ is also an element and meams ``$i$ explicitly knows that imaginary $j$ with the awareness of $(i)$ explicitly knows $p$,'' which is ``$i$ explicitly knows that $j$ explicitly knows $p$.'' By contrast, $E_{(i)} E_{(j,k)} p$ is not an element of the language. It means ``$i$ explicitly knows that imaginary $k$ with the awareness of $(j)$ explicitly knows $p$.'' Although what follows the first ``that'' should unfold in $i$'s mind, it unfolds only in $j$'s mind. Our syntax is designed to remove such strange propositions.

\subsection{Semantics}
We move on to define the semantics. The semantics are given in the Kripke-style. It has the standard Kriple model for knowledge and is based on the awareness structure \cite{fagin1988belief}. We extend an awareness set and the equivalence relation in \cite{yudai2022-1} with a chain of belief for awareness and incorporate to a model.
\begin{df}
  An \textit{epistemic model with awareness} $M$ is a tuple $\langle W, \{\mathop{\sim_i}\}_{i\in\mathcal{G}}, V,\{\mathscr{A}_{\theta}\}_{\theta\in\Theta}\rangle$, where:
     \begin{itemize}
      \item[$\bullet$] $W \text{ is a non-empty set of possible worlds}$;
      \item[$\bullet$] $\mathop{\sim_i}
      \text{ is an equivalence relation on }W$;
      \item[$\bullet$] $V : \mathcal{P} \to 2^{W}$;
      \item[$\bullet$] $\mathscr{A}_{\theta}\text{ is a non-empty set 
      of atomic propositions satisfying the condition:}$
      \begin{itemize}
        \item If $\theta'\preceq \theta$, 
        then $\mathscr{A}_{\theta} \subseteq 
        \mathscr{A}_{\theta'}$ for every $\theta'\in \Theta$.
      \end{itemize}    
     \end{itemize}
   \end{df}
\noindent
The part of $\langle W, \{\mathop{\sim_i}\}_{i\in\mathcal{G}}, V\rangle$ is a Kripke model that corresponding to $\mathbf{S5}$. A set $\mathscr{A}_{\theta}$ is called an awareness set of $\theta$. Let $\theta$ be $(i_1,\dots,i_n)$, and $\mathscr{A}_{(i_1,\dots,i_n)}$ represents a set of propositions of which  $i_1$ believes $\cdots$ $i_{(n-1)}$ believes that the last member $i_n$ is aware. The condition in the definition of an awareness set means that if a chain $\theta$ is greater than equal to another chain $\theta'$, then the awareness set of $\theta'$ is included in the awareness set of $\theta$. This reflects an intuition that $i$ should be aware of all the propositions of which $i$ believes that $j$ is aware or imaginary $j$ in $i$'s mind is aware.
As for the pair $\mathscr{A}_{(i,i)}$ and $\mathscr{A}_{(i\cdot\theta\cdot j)}$, we may need to supplementary explain. In this case, $\mathscr{A}_{(i\cdot\theta\cdot j)}\subseteq\mathscr{A}_{(i,i)}$ holds because $(i)\simeq (i,i) \prec (i\cdot\theta\cdot j)$ by the definition of the order. In the case of a symmetric order between $\theta$ and $\theta'$, the two awareness sets of $\theta$ or $\theta'$ are identical. The pair $(M,w)$ with $M$ and $w\in W$ is called a pointed model.



We define another equivalence relation to formalize the effect of an agent's awareness on her own reasoning. The relation is given depending on an epistemic model with awareness.  
\begin{df}
  For each $\theta$, a relation $\mathop{\approx_{\theta}}$ on $W$ 
  is defined by $(w,v) \in \mathop{\approx_{\theta}} \textit{\  iff,  \ } 
  w \in V(p) \textit{\ iff \ } v \in V(p) \text{ for every } p \in \mathscr{A}_{\theta}$.
\end{df}
\noindent
The relation $\mathop{\approx_{\theta}}$ is called the \textit{indistinguishable relation} of $\theta$. The relation occurs between possible worlds with different valuations only for unaware propositions. It means the relation appears or disappears in conjunction with an awareness set. It is one of the properties between the relation and an awareness set that $\mathscr{A}_{\theta}\subseteq \mathscr{A}_{\theta'}$ implies $\mathop{\approx_{\theta'}}\subseteq \mathop{\approx_{\theta}}$. The last member with the awareness of $\theta$ cannot distinguish such possible worlds because of her lack of awareness. This relation is an equivalence relation and constructs partitions on $W$. Possible worlds collapsed by the relation as one equivalence class can be interpreted as one world for the last member with the awareness of $\theta$. 

Next, we move on to define the satisfaction relation. Before the definition, we introduce some notations: a set $At(\varphi)$ denotes the set of atomic propositions that occur in $\varphi$; a relation $\mathop{\sim_i} \circ\mathop{\approx_{\theta}}$ denotes a sequential composition of $\approx_{\theta}$ and $\mathop{\sim_i}$; for a binary relation $R$, a relation $R^+$ denotes the transitive closure of $R$. The closure $R^+$ is the smallest set such that $R\subseteq R^+$, and for all $x, y, z$, if $(x,y)\in R^+$ and $(y,z)\in R^+$, then $(x,z)\in R^+$.
\begin{df}
  For any epistemic models with awareness $M$ and possible worlds $w \in W$, a satisfaction relation $\vDash$ is given as follows: 
  \begin{align*}
    M,w \vDash p &\textit{\  iff  \ } w \in V(p) ;
    \\[-3pt] M,w \vDash \neg \varphi &\textit{\  iff  \ } M,w 
     \nvDash\varphi;\\[-3pt]
     M,w \vDash \varphi\wedge\psi &\textit{\  iff  \ } 
     M,w\vDash\varphi \text{, and } M,w\vDash\psi ; 
    \\[-3pt] M,w \vDash A_{\theta} 
     \varphi &\textit{\  iff  \ } \mathrm{At}(\varphi) 
     \subseteq \mathscr{A}_{\theta};\\[-3pt]
    M,w \vDash I_i\varphi &\textit{\  iff  \ } 
    M,v\vDash \varphi \text{ for all } v \text{ such that }
    (w,v)\in \mathop{\sim_i};\\[-3pt]
    M,w\vDash [\approx]_{\theta}\varphi &\textit{\  iff  \ } 
    M,v\vDash\varphi \text{ for all } v \text{ such that }
    (w,v)\in \mathop{\approx_{\theta}};\\[-3pt]
    M,w\vDash C_{\theta}\varphi &\textit{\  iff  \ } 
    M,v\vDash \varphi \text{ for all } v
    \text{ such that } (w,v) \in (\mathop{\sim_i} \circ\mathop{\approx_{\theta}})^+
    ;\\[-3pt]
     M,w \vDash E_{\theta} \varphi &\textit{\  iff  \ } 
     M,w \vDash A_{\theta}\varphi  \text{, and }
     M,w\vDash C_{\theta}\varphi,     
  \end{align*}
\end{df}
\noindent
From Definition 1, if both an indistinguishable relation $\mathop{\approx_{\theta}}$ and an accessibility relation $\mathop{\sim_i}$ is an equivalence relation, $( \mathop{\sim_i} \circ\mathop{\approx_{\theta}})^+$ is also an equivalence relation. This is because the reverse direction on the composition is also reachable, although it consumes a few extra steps. Thus, $( \mathop{\sim_i} \circ\mathop{\approx_{\theta}})^+$ gives a new partition of possible worlds. This partition corresponds to imaginary $i$'s reasoning limited by the awareness of $\theta$. We can formalize knowledge taking into account the propositions of which the agent is aware. We can also find that $[\approx]_{\theta} I_i\varphi$ corresponds to $\mathop{\sim_i} \circ\mathop{\approx_{\theta}}$. In contrast, this relation is not equivalent, unlike its transitive closure.

The validity is defined in the usual way.  
\begin{df}
  A formula $\varphi$ is valid at $M$, 
  if $\varphi$ holds at every pointed model $M,w$ in $M$, 
  which is denoted by $M\vDash\varphi$. 
  A formula $\varphi$ is valid
  if $\varphi$ holds at every pointed model $M,w$, 
  which is denoted by $\vDash \varphi$. 
\end{df}

Note that there are local and global definitions of an awareness set. The former defines an awareness function that takes a possible world as an argument and changes elements of an awareness set for each possible world. The latter defines an awareness set as the same in all possible worlds. Generally, a state of an agent's awareness should be fixed within a set of possible worlds that she considers. The global definition is used in the logic that does not consider the outside of a specific agent's accessible worlds, such as a single-agent case. On the other hand, a local definition is used to represent possible worlds outside of an agent's accessible worlds. It is possible to express that a state of an agent's awareness differs between the agent and another agent.

Our logic adopts the global one. By incorporating the notion of a chain of belief for awareness, it is possible to represent the advantage of the local definition: to express a state of an agent's awareness that another agent thinks may differ from that of the agent's actual awareness. This is because a state of an agent's awareness is separately set for each agent. However, another feature of the local one is not represented in our logic. A situation where an agent does not have any belief at all may be interesting. Knowledge in such a situation is only for the local one. In order to consider such a situation, we can modify the definition of an awareness set into an awareness function in our logic.

\subsection{Properties}
In this section, we discuss some properties of our logic.

As our logic treats ``awareness of,'' the property of ``awareness of'' is satisfied \cite{van2015handbook}: if an agent is aware of atomic propositions, she is aware of more complex formulas produced by the atomic propositions. This means if an agent is aware of $\varphi$, that is, she can refer to $\varphi$, then she is also aware of $\neg\varphi$. 
\begin{itemize}
  \item[$\bullet$] $\vDash A_{\theta}\varphi \leftrightarrow 
  A_{\theta}\neg\varphi$, 
  \item[$\bullet$] $\vDash A_{\theta}(\varphi\wedge\psi) \leftrightarrow 
  A_{\theta}\varphi \wedge A_{\theta}\psi$, 
  \item[$\bullet$] $\vDash A_{\theta}\varphi \leftrightarrow A_{\theta} 
  A_{\theta'}\varphi$,
  \item[$\bullet$] $\vDash A_{\theta}\varphi \leftrightarrow A_{\theta} 
  [\approx]_{\theta'}\varphi$,
  \item[$\bullet$] $\vDash A_{\theta}\varphi \leftrightarrow 
  A_{\theta} C_{\theta'}\varphi$,
  \item[$\bullet$] $\vDash A_{\theta} \varphi\leftrightarrow A_{\theta} 
  I_i \varphi$,
  \item[$\bullet$] $\vDash A_{\theta} \varphi \leftrightarrow 
  A_{\theta} K_{\theta'} \varphi$.
\end{itemize}

Because of the definition of an awareness set, when a sequential iteration of an agent occurs in a chain, the awareness set of the chain is the same as the awareness set of the chain where the iterated part is removed. 
\begin{itemize}
  \item[$\bullet$] 
  $\vDash A_{(i,j,j)} \varphi 
  \leftrightarrow A_{(i,j)} \varphi$.
  \item[$\bullet$] 
  $\vDash [\approx]_{(i,j,j)} \varphi 
  \leftrightarrow [\approx]_{(i,j)} \varphi$.
  \item[$\bullet$] 
  $\vDash E_{(i,j,j)} \varphi 
  \leftrightarrow E_{(i,j)} \varphi$.
\end{itemize}
The third formula means that explicit knowledge of $j$ with the awareness that $i$ believes $j$ believes that $j$ has equals to explicit knowledge of $j$ with the awareness that $i$ believes that $j$ has.


We can also say: 
\begin{itemize}
  \item[$\bullet$] 
  $\vDash A_{(i,j,k)}\varphi \rightarrow A_{(i,j)} \varphi$,
  \item[$\bullet$] 
  $\vDash [\approx]_{(i,j,k)} \varphi \to [\approx]_{(i,j)}\varphi$,
  \item[$\bullet$] 
  $\nvDash E_{(i,j,k)}\varphi \rightarrow E_{(i,j)} \varphi$.
\end{itemize}
The first and second expressions are natural. This is because the propositions of which $i$ believes $j$ believes that $k$ is aware necessarily are those of which $i$ believes that $j$ is aware.
In contrast, the third formula is not valid. Accessibility relations do not change in conjunction with awareness sets. Even if the indistinguishable relation of $(i,j,k)$ includes indistinguishable of $(i,j)$, the formula can be false at some worlds accessible on $j$'s accessibility relation from the evaluate world.

In logic of awareness, it is regarded that the negative introspection for unawareness holds, which means an agent never knows any proposition of which the agent is not aware. Our logic has the property because explicit knowledge of a proposition requires awareness of the proposition.
\begin{itemize}
  \item[$\bullet$] $\vDash \neg E_{\theta} \neg A_{\theta}\varphi$.
\end{itemize}
As for implicit knowledge, the formula $L_{\theta} \neg A_{\theta}\varphi$ may be true.

To understand nested knowledge in our logic, we also give some formulas.
\begin{itemize}
  \item[$\bullet$] 
  $\nvDash E_{(i,j)}\varphi 
  \rightarrow E_{(i)} E_{(i,j)}\varphi$,
  \item[$\bullet$] 
  $\vDash E_{(i)} E_{(i,j)}\varphi 
  \rightarrow E_{(i,j)}\varphi$.
\end{itemize}
This means that $i$'s explicit knowledge that imaginary $j$ with the awareness of $(i,j)$ explicitly knows $\varphi$ knows coincides with the explicit knowledge of imaginary $j$ with the awareness that $i$ believes that $j$ has. Our logic distinguishes the sentence that $i$ explicitly knows that $j$ explicitly knows $\varphi$ from the sentence that $j$ with the awareness that $i$ believes $j$ has explicitly knows $\varphi$.

As stated in Section 1, our explicit knowledge operator has properties of knowledge in $\mathrm{EL}$. 
\begin{itemize}
  \item[$\bullet$] 
  $\vDash E_{(i)} E_{(i,j)}\varphi \rightarrow E_{(i,j)}\varphi$,
  \item[$\bullet$] 
  $\nvDash E_{(i)} E_{(i,j)}\varphi \rightarrow E_{(j)}\varphi$,
  \item[$\bullet$] 
  $\vDash \neg E_{(i)}\varphi \rightarrow \neg E_{(i,j,i)} \varphi$.
\end{itemize}
The first formula is $\mathrm{T}$ axiom in modal logic, which is a property of knowledge in $\mathrm{EL}$. The second formula does not hold because if $i$ explicitly knows that $j$ explicitly knows $\varphi$, that does not mean that $j$ herself explicitly knows $\varphi$. However, that means that imaginary $j$ with the awareness in $i$'s belief in $j$'s belief explicitly knows $\varphi$ as the first expression shows. The third formula means if $a$ does not explicitly know $p$, then an imaginary $a$ with the state of $a$'s awareness in $b$'s mind in $a$'s mind does not explicitly know $p$. The set of imaginary the state of $a$'s awareness in $b$'s mind in $a$'s mind should be smaller than the set of what $a$ herself is aware of. 

As for implicit knowledge, it has the same properties as the necessity operator in modal logic $\mathbf{S5}$, such as $I_i\varphi\to \varphi$, $I_i\varphi\to I_i I_i \varphi$, and $\neg I_i\varphi\to I_i \neg I_i \varphi$.
The relation between $I_i$ operator and $C_{\theta}$ operator stated in syntax is represented as $C_{\theta}\varphi \to I_i\varphi$. 

\subsection{Formalization of The Example}
This subsection gives the formalization of the example introduced in Section 1 with our logic. Let $\Theta$ be $\{(a),(b),(a,b),(b,a),(a,b,a)\}$. In the example, only $a$ is aware of all the atomic propositions. Formally, we write $\mathscr{A}_{(a)}=\{p_a,p_b,q\}$, and $\mathscr{A}_{(b)}=\{p_a,p_b\}$. Owner $a$ believes that $b$ is also aware of the wholesaler, which is formalized as $\mathscr{A}_{(a,b)}=\{p_a,p_b,q\}$. By contrast, $b$ believes that $a$ is not aware of the wholesaler as well as herself, which is formalized as $\mathscr{A}_{(b,a)}=\{p_a,p_b\}$. Finally, $a$ believes $b$ believes that $a$ is not aware of the wholesaler, which is formalized as $\mathscr{A}_{(a,b,a)}=\{p_a,p_b\}$. Then, there are indistinguishable relations $\mathop{\approx_{(b)}}$, $\mathop{\approx_{(b,a)}}$, and $\mathop{\approx_{(a,b,a)}}$ between possible worlds with different valuation only for $q$. Formally, $\mathop{\approx_{(b)}} =\mathop{\approx_{(b,a)}} = \mathop{\approx_{(a,b,a)}} = \{(w_1,w_2),(w_2,w_1),(w_1,w_1),\dots,(w_5,w_5)\}$ in Figure 2.  

In Figure 2, the equivalence classes of the indistinguishable relation $\mathop{\approx_{(a,b,a)}}$, which is the same as $\mathop{\approx_{(b)}}$ and $\mathop{\approx_{(b,a)}}$, are represented by the light grey background. By interpreting the equivalence class as one possible world for $a$ with the awareness of $(a,b,a)$, the figure represents the same graph as the Kripke model that is only for $p_a$ and $p_b$ (Figure 3), in terms of possible worlds and accessibility relations.

In the model $M$ in Figure 2, all the formulas are consistent with each knowledge stated in Section 1. Among those, we focus on $a$'s nested knowledge here, which is $1: E_{(a)} q\to E_{(a)} p_b$, and $2: E_{(a)} E_{(a,b)} \neg E_{(a,b,a)} p_b$. This knowledge is seemingly contradictory, but they correctly represent the knowledge of agents with a different awareness.
\begin{itemize}
  \item[$\bullet$] $M\vDash E_{(a)} q\to E_{(a)} p_b$ 
  \item[$\bullet$] $M\vDash E_{(a)} E_{(a,b)} \neg E_{(a,b,a)} p_b$
\end{itemize}
As for $1$, it is vacuously true in the other worlds than $w_1$. Since $p_b\in\mathcal{A}_(a)$, $w_1$ is only world reachable on $\sim_{(a)}$ and $\approx_{(a)}$ from $w_1$, and $p_b$ is true in $w_1$, then $E_{(a)} q\to E_{(a)} p_b$ is true for all worlds. As for $2$, the formula in the other worlds than $w_1$ requires that $p_b$ is true in all the worlds because all the worlds are reachable on three relations $\sim_a$, $\sim_b$, $\approx_{(a,b,a)}$. Moreover, $w_3$, where $p_b$ is false, is reachable on $\approx_{(a,b,a)}$ from $w_1$. Thus, $E_{(a)} E_{(a,b)} \neg E_{(a,b,a)} p_b$ is true for all worlds.

Our logic $\mathcal{ALPC}$ correctly represented both owners' explicit knowledge including nested one in the example. We can say that this logic is one of formalizations to represent and analyze nested explicit knowledge.

\begin{figure}[h]
  \begin{tabular}{cc}
  \begin{minipage}{0.5\hsize}
  \centering
  \scalebox{0.9}[0.9]{
  \begin{tikzpicture}
  \fill[lightgray](-1.25,-1.25)rectangle(4.25,1.25);
  \fill[lightgray](4.75,-1.25)rectangle(7.25,1.25);
  \fill[lightgray](1.75,-1.75)rectangle(4.25,-4.25);
  \fill[lightgray](4.75,-1.75)rectangle(7.25,-4.25);
  \draw(0,0)circle(1);
  \draw(0.95,-0.95)node{{\large $w_1$}};
  \draw(3,0)circle(1);
  \draw(3.95,-0.95)node{{\large $w_2$}};
  \draw(6,0)circle(1);
  \draw(6.95,-0.95)node{{\large $w_3$}};
  \draw(3,-3)circle(1);
  \draw(3.95,-3.95)node{{\large $w_4$}};
  \draw(6,-3)circle(1);
  \draw(6.95,-3.95)node{{\large $w_5$}};
  \draw(0,0.1)node{{\Large $p_a $ $p_b$}};
  \draw(0,-0.05)node[below]{{\Large $ q$}};
  \draw(3,0.1)node{{\Large $p_a $ $p_b$}};
  \draw(3,-0.05)node[below]{{\Large $\neg q$}};
  \draw(6,0.1)node{{\Large $p_a $ $\neg p_b$}};
  \draw(6,-0.05)node[below]{{\Large $\neg q$}};
  \draw(3,-2.9)node{{\Large $\neg p_a$ $ p_b$}};
  \draw(3,-3.05)node[below]{{\Large $\neg q$}};
  \draw(6,-2.9)node{{\Large $\neg p_a $ $\neg p_b$}};
  \draw(6,-3.05)node[below]{{\Large $\neg q$}};
  \draw[thick, <->](4,0)--(5,0);
  \draw(4.5,0)node[fill=white]{{\large $a$}};
  \draw[thick, <->](4,-3)--(5,-3);
  \draw(4.5,-3)node[fill=white]{{\large $a$}};
  \draw[thick, <->](3,-1)--(3,-2);
  \draw(3,-1.5)node[fill=white]{{\large $b$}};
  \draw[thick, <->](6,-1)--(6,-2);
  \draw(6,-1.5)node[fill=white]{{\large $b$}};
  \end{tikzpicture}\par}
  \caption{A two multi-agent model where a set of possible worlds is $\{w_1,w_2,w_3,w_4,w_5\}$. Reflexive arrows are omitted. Grey backgrounds refer to equivalence classes of $\mathop{\approx_{(a,b,a)}}$.}
  \end{minipage}
  \begin{minipage}{0.5\hsize}
    \centering
    \scalebox{0.9}[0.9]{
    \begin{tikzpicture}
    \draw(3,0)circle(1);
    \draw(3.95,-0.95)node{{\large $w_1'$}};
    \draw(6,0)circle(1);
    \draw(6.95,-0.95)node{{\large $w_2'$}};
    \draw(3,-3)circle(1);
    \draw(3.95,-3.95)node{{\large $w_3'$}};
    \draw(6,-3)circle(1);
    \draw(6.95,-3.95)node{{\large $w_4'$}};
    \draw(3,0)node{{\Large $p_a $ $p_b$}};
    \draw(6,0)node{{\Large $p_a $ $\neg p_b$}};
    \draw(3,-3)node{{\Large $\neg p_a$ $ p_b$}};
    \draw(6,-3)node{{\Large $\neg p_a $ $\neg p_b$}};
    \draw[thick, <->](4,0)--(5,0);
    \draw(4.5,0)node[fill=white]{{\large $a$}};
    \draw[thick, <->](4,-3)--(5,-3);
    \draw(4.5,-3)node[fill=white]{{\large $a$}};
    \draw[thick, <->](3,-1)--(3,-2);
    \draw(3,-1.5)node[fill=white]{{\large $b$}};
    \draw[thick, <->](6,-1)--(6,-2);
    \draw(6,-1.5)node[fill=white]{{\large $b$}};
    \end{tikzpicture}\par}
    \caption{A model only concering $p_a$ and $p_b$}
  \end{minipage}
  \end{tabular}
  \end{figure}

\section{Hilbert System of $\mathbf{ALPC}$}
We move on to the proof theory for our logic. The Hilbert system of $\mathbf{ALPC}$ is given in Table 1. Axiom $\mathrm{AN}$, $\mathrm{ACN}$, $\mathrm{AA}$, $\mathrm{AL}$, $\mathrm{A[\approx]}$, $\mathrm{ACM}$, and $\mathrm{AK}$ mean that if an agent is aware of atomic propositions, the agent is aware of more complex formulas produced by the atomic propositions and correspond to the notion of ``awareness of.'' Axiom $\mathrm{AI}$ corresponds to the condition in the definition of an awareness set. Axiom $\mathrm{KA}$ and $\mathrm{NKA}$ reflect the global definition of an awareness set. For $\mathrm{K_L, T_L, 5_L, K_{[\approx]}, T_{[\approx]}},$ and $\mathrm{5_{[\approx]}}$, we adopt $\mathrm{K}$,$\mathrm{T}$, and $\mathrm{5}$ in modal logic. $\mathrm{5_L}$ called negative introspection in $\mathrm{EL}$ means that an agent always knows what the agent does not know. This axiom also characterizes logical omniscience. In our logic, as with most logics of awareness, this formula does not hold for an explicit knowledge operator. Instead, $\neg E_{\theta}\varphi\wedge A_{\theta}\neg E_{\theta}\varphi\to E_{\theta} \neg E_{\theta}\varphi$ is valid. $\mathrm{K_C}, \mathrm{IND}$, and $\mathrm{MIX}$ are based on axioms of logic with common knowledge because both logics use a transitive closure. $\mathrm{KAC}$ corresponds to the constitution of $E_{\theta}\varphi$ on the satisfaction relation, which is explicit knowledge is the knowledge that meets a kind of implicit knowledge represented by $C_{\theta}$ and an aware proposition. 

\begin{df}
  A system $\mathbf{ALPC}$ is a set of formulas that contains the axioms in Table 1 and is closed under inference rules in the table. We write $\vdash\varphi$ if $\varphi\in$ $\mathbf{ALPC}$. Let $\Gamma$ be a set of formulas in $\mathcal{L}_{(\mathcal{P,G},\Theta)}$ and $\bigwedge\Gamma$ be an abbreviation of $\bigwedge_{\varphi\in \Gamma}\varphi$. If there is a finite subset $\Gamma'$ of $\Gamma$ such that $\vdash\bigwedge\Gamma'\to \varphi$, we write $\Gamma\vdash\varphi$ and call it that $\varphi$ is deducible from $\Gamma$.
\end{df}

\begin{table}[h]
  \caption{Table 1: Axiom schemas and inference rules of \textbf{ALPC}}	
  \centering
   \begin{tabular}{|l|l|}
    \hline
    \multicolumn{2}{|c|}{Axiom schemata}\\\hline
    $\mathrm{TAUT}$ & The set of propositional tautologies\\
        $\mathrm{AN}$ & $A_{\theta}\varphi \leftrightarrow A_{\theta}\neg\varphi$\\
        $\mathrm{ACN}$ & $A_{\theta}(\varphi\wedge\psi) \leftrightarrow 
        A_{\theta}\varphi \wedge A_{\theta}\psi$\\
        $\mathrm{AA}$ & $A_{\theta}\varphi \leftrightarrow A_{\theta} 
        A^{\theta'}_j\varphi$\\
        $\mathrm{AL}$ & $A_{\theta} \varphi\leftrightarrow A_{\theta} 
        I_j \varphi$\\
        $\mathrm{A[\approx]}$ & 
        $A_{\theta}\varphi \leftrightarrow A_{\theta} 
        [\approx]^{\theta'}_j\varphi$\\
        $\mathrm{ACM}$ & $A_{\theta}\varphi \leftrightarrow 
        A_{\theta} C^{\theta'}_j\varphi$\\
        $\mathrm{AK}$ & $A_{\theta} \varphi \leftrightarrow 
        A_{\theta} K^{\theta'}_j \varphi$\\
        $\mathrm{AN[\approx]}$ & $A_{\theta} p\wedge p
         \to [\approx]_{\theta} p$\\
        $\mathrm{AI}$ & $A_{\theta} \varphi
        \to A_{\theta'} \varphi$\\
        $\mathrm{KA}$ & $A_{\theta} \varphi
        \to \bigwedge_{\theta'\in \Theta} C_{\theta'} A_{\theta} \varphi$\\
        $\mathrm{NKA}$ & $\neg A_{\theta} \varphi
        \to \bigwedge_{\theta'\in \Theta} C_{\theta'} \neg A_{\theta} \varphi$\\
        $\mathrm{K_L}$ & 
        $I_i(\varphi\to \psi)\to 
        (I_i\varphi \to I_i\psi)$\\
        $\mathrm{T_L}$ & $I_i \varphi \to \varphi$\\
        $\mathrm{5_L}$ & $\neg I_i\varphi
        \to I_i \neg I_i\varphi$\\
        $\mathrm{K_{[\approx]}}$ & 
        $[\approx]_{\theta}(\varphi\to \psi)\to 
        ([\approx]_{\theta}\varphi \to [\approx]_{\theta}\psi)$\\
        $\mathrm{T_{[\approx]}}$ & $[\approx]_{\theta} \varphi \to \varphi$\\
       $\mathrm{5_{[\approx]}}$ & $\neg [\approx]_{\theta}\varphi
        \to [\approx]_{\theta} \neg [\approx]_{\theta}\varphi$\\
        $\mathrm{K_{C}}$ & $C_{\theta}(\varphi\to \psi)\to 
        (C_{\theta}\varphi \to C_{\theta}\psi)$\\
        $\mathrm{MIX}$ & $C_{\theta}\varphi \to \varphi
        \wedge[\approx]_{\theta} I_i C_{\theta}\varphi$\\
        $\mathrm{IND}$ & 
        $C_{\theta}(\varphi\to [\approx]_{\theta} I_i\varphi
        )\to
        (\varphi \to C_{\theta}\varphi)$\\
        $\mathrm{KAC}$ & $E_{\theta}\varphi\leftrightarrow A_{\theta}\varphi\wedge 
        C_{\theta}\varphi$\\\hline
        \multicolumn{2}{|c|}{Inference Rules}\\\hline
        $\mathrm{MP}$ & If $\vdash \varphi$ and $\vdash \varphi\to\psi$, then 
        $\vdash \psi$\\
        $\mathrm{LG}$ & If $\vdash \varphi$ then $\vdash I_i\varphi$\\
        $\mathrm{[\approx] G}$ &
        If $\vdash \varphi$ then $\vdash [\approx]_{\theta}\varphi$\\
        $\mathrm{CG}$ & If $\vdash \varphi$ 
        then $\vdash C_{\theta}\varphi$\\\hline
  \end{tabular}
\end{table}

\subsection{Soundness}
We prove that all the theorems are valid, which is the soundness of the logic.  
\begin{thm}
  If $\vdash\varphi$, then $\vDash\varphi$.
  \end{thm}
  \begin{proof} 
    We prove it for any formulas by induction on the structure of $\mathbf{ALPC}$. First, we prove that all the axioms are valid. For $I_i$ and $[\approx]_{\theta}$, it is proven in a similar way in \textbf{S5} because both accessibility relations and indistinguishable relations are equivalence relations. For $A_{\theta}$, we show only $\mathrm{AN}$ and $\mathrm{AN[\approx]}$ here. Axioms $\mathrm{ACN}$, $\mathrm{AA}$, $\mathrm{AL}$, $\mathrm{A[\approx]}$, $\mathrm{ACM}$, and $\mathrm{AK}$ are proven in the same way as these two axioms. For $\mathrm{AI}$, $\mathrm{KA}$, and $\mathrm{NKA}$, it is proven quickly from the definition. For $\mathrm{KAC}$, the proof is trivial.
    \begin{itemize}
    \item For $\mathrm{AN}$, suppose that $M,w\vDash A_{\theta}\varphi$, then $\mathrm{At}(\varphi)\subseteq \mathscr{A}_{\theta}$.
    Since $\mathrm{At}(\varphi) = \mathrm{At}(\neg\varphi)$, $\mathrm{At}(\neg\varphi)\subseteq \mathscr{A}_{\theta}$.
    Thus, $M,w\vDash A_{\theta}\neg\varphi$.
    The reverse direction is proven similarly.
    \item For $\mathrm{AN[\approx]}$, suppose that $M,w\vDash A_{\theta} p 
    \wedge p$, then $p \in \mathscr{A}_{\theta}$ and $w\in V(p)$.
     We consider an arbitary $v$ such that for every atomic proposition $q\in \mathscr{A}_{\theta}$, $w\in V(q)$ if and only if $v\in V(q)$.
     From the assumptions, $M,v\vDash p$. Thus, $M,w\vDash [\approx]_{\theta} p$.
    \item For $\mathrm{K_C}$, suppose that $M,w\vDash  C_{\theta}(\varphi\to\psi)$ and $M,w\vDash C_{\theta}\varphi$. For all $v$ such that $(w,v)\in (\mathop{\sim_i}\circ\mathop{\approx_{\theta}})^+$, $M,v\vDash\varphi\to\psi$ and $M,v\vDash\varphi$. $M,v\vDash\psi$. Thus, $M,w\vDash C_{\theta}\psi$.
    \item For $\mathrm{MIX}$, suppose that $M,w\vDash  C_{\theta}\varphi$. 
    Since $(\mathop{\sim_i} \circ\mathop{\approx_{\theta}})^+$ 
    is equivalent and the transitive closure, $M,w\vDash\varphi\wedge
    [\approx]_{\theta} I_i C_{\theta}\varphi$. 
    \item For $\mathrm{IND}$, suppose that $M,w\vDash 
    C_{\theta}(\varphi\to [\approx]_{\theta} I_i\varphi)$ and $M,w\vDash\varphi$, then for all $v$ such that $(w,v)\in (\mathop{\sim_i} \circ
    \mathop{\approx_{\theta}})^+$, $M,v\vDash\varphi\to[\approx]_{\theta} I_i\varphi$. Thus, $M,w\vDash[\approx]_{\theta} I_i\varphi$. 
    $\varphi$ holds at all the possible worlds 
    from $w$ on $\mathop{\sim_i} \circ\mathop{\approx_{\theta}}$, 
    and $[\approx]_{\theta} I_i\varphi$ holds even at that worlds. Therefore, $M,w\vDash C_{\theta}\varphi$.
    \end{itemize}
    The remaining task is to prove that if the assumptions are valid, they are also valid for all the inference rules. All of them are straightforward.
  \end{proof}  
  
\subsection{Completeness}
We prove the reverse direction: all the valid formulas are the theorems. We use the idea of a canonical model in the proof on modal logic \cite{chellas1980modal}. In \textbf{ALPC}, we can take a set of formulas, such as $\Phi = \{([\approx]_{\theta} I_i)^n \varphi\mid n\in \mathbb{N}\} \cup \{\neg C_{\theta}\varphi\}$ for each $\theta\in\Theta$, where $([\approx]_{\theta} I_i)^n$ is $n$ iterations of $[\approx]_{\theta} I_i$. Therefore, our logic is no longer compact. It is necessary to restrict a canonical model to a specific set of formulas. This technique is used in the proof on logic with common knowledge defined by the reflexive-transitive closure of relations. We customize the tools and techniques in \cite{van2007dynamic} for our logic in a similar to \cite{yudai2022-1}. Moreover, we cannot construct a unique canonical model, unlike the standard strategy on modal logic. Because of the global definition, the truth value of $A_{\theta}\varphi$ is the same within the domain in a model, and then we construct multiple models with different domains based on the idea of a canonical model, called \textit{divided models}.

It is necessary to restrict a maximal consistent set used in constructing a divided model to a specific set of formulas. At first, we define a \textit{closure} as a restricted set of formulas.
\begin{df}
  Let $cl : \mathcal{L}\to 2^{\mathcal{L}}$ be the function such that for every $\varphi\in\mathcal{L}$, $cl(\varphi)$ is the smallest set satisfying that: 
  \begin{itemize}
    \item[1.] $\varphi\in cl(\varphi)$; 
    \item[2.] If $\psi\in cl(\varphi)$ then $sub(\psi)\subseteq cl(\varphi)$, where $sub(\psi)$ is the set of subformulas of $\psi$;
    \item[3.] If $\psi\in cl(\varphi)$ and $\psi$ is not a form of negation, then $\neg\psi\in cl(\varphi)$;
    \item[4.] If $A_{\theta} \psi \in cl(\varphi)$, then $C_{\theta'}A_{\theta} \psi, C_{\theta'}\neg A_{\theta} \psi\in cl(\varphi)$ for each chain $\theta'\in \Theta$;
    \item[5.] If $A_{\theta} \psi \in cl(\varphi)$, then $[\approx]_{\theta} p, A_{\theta}\chi, A_{\theta'}\psi\in cl(\varphi)$ for each chain $\theta'\in \Theta$, where $\chi\in sub(\psi)$ and $p$ is an atomic proposition in $cl(\varphi)$;
    \item[6.] If $I_i\psi\in sub(\varphi)$, then $ I_i I_i\psi$, $I_i\neg I_i\psi\in cl(\varphi)$;
    \item[7.] If $[\approx]_{\theta}\psi\in sub(\varphi)$, then $[\approx]_{\theta}[\approx]_{\theta}\psi$, $[\approx]_{\theta}\neg[\approx]_{\theta}\psi\in cl(\varphi)$;
    \item[8.] If $C_{\theta}\psi\in cl(\varphi)$, then $[\approx]_{\theta} I_i C_{\theta}\psi\in cl(\varphi)$;
    \item[9.] If $I_i C_{\theta} \psi\in cl(\varphi)$, then $I_i I_i C_{\theta} \psi$, $I_i\neg I_i C_{\theta}\psi\in cl(\varphi)$;
    \item[10.] If $[\approx]_{\theta} I_i C_{\theta}\psi\in cl(\varphi)$, then $ [\approx]_{\theta}[\approx]_{\theta} I_iC_{\theta}\psi$ ,$[\approx]_{\theta}\neg[\approx]_{\theta} I_iC_{\theta}\psi\in cl(\varphi)$;
    \item[11.] If $E_{\theta}\psi\in cl(\varphi)$, then $A_{\theta} \psi$, $C_{\theta} \psi\in cl(\varphi)$.
  \end{itemize} 
\end{df}
\noindent

\begin{lemma}
  For every $\varphi$, $cl(\varphi)$ is finite.
\end{lemma}
\begin{proof}
  We prove it by induction on the structure of $cl(\varphi)$. It is easily provable for all the conditions because a set of agent $\mathcal{P}$ and a set of chains $\Theta$ is finite.
\end{proof}

We define a consistent set and a maximal consistent set in a closure. A set of the latter is used as a domain of a divided model.
\begin{df}
  Let $\Phi$ be the closure of a formula. 
  $\Gamma$ is a consistent set in $\Phi$ iff 
  \begin{itemize}
    \item[$\bullet$] $\Gamma \subseteq \Phi$, 
    \item[$\bullet$] $\Gamma \nvdash \bot$.
  \end{itemize}
  Moreover, $\Gamma$ is a maximal consistent set in $\Phi$ iff 
  \begin{itemize}
    \item[$\bullet$] $\Gamma$ is a consistent set,
    \item[$\bullet$] There is no $\Gamma'\subseteq\Phi$ such that $\Gamma\subset\Gamma'$ and $\Gamma'\nvdash \bot$.
  \end{itemize}
\end{df}

We prove that a consistent set can always expand a maximal consistent set that includes the original set. 
\begin{lemma}
  Let $\Phi$ be the closure of a formula. If $\Gamma$ is a consistent set in $\Phi$, then there exists a maximal consistent set $\Delta$ in $\Phi$ such that $\Gamma\subseteq\Delta$. 
\end{lemma}
\begin{proof}
  It is proven in the same way as Lindenbaum's lemma. We obtain a maximal consistent set by adding formulas in $\Phi$ to $\Gamma$ so as to preserve the consistency.
\end{proof}

Furthermore, we prove a corollary of Lemma 2.
\begin{lemma}
  Let $\Phi$ be the closure of a formula. For every $\varphi\in\Phi$, $\Gamma\vdash\varphi$ iff $\varphi\in\Delta$ for every maximal consistent set $\Delta$ in $\Phi$ such that $\Gamma\subseteq\Delta$.
\end{lemma}
\begin{proof}
  From left to right, suppose that $\Gamma\vdash\varphi$, $\Gamma\subseteq\Delta$, and $\Delta\in W^*_{\Lambda}$. There exists a finite subset $\Gamma'\subseteq\Gamma$ such that $\vdash \bigwedge\Gamma'\to \varphi$. Since $\Gamma\subseteq\Delta$ and $\Delta$ is a maximal consistent set, $\varphi\in\Delta$.
  From right to left, we prove it by contraposition. It is sufficient to prove that if $\Gamma\nvdash\varphi$, then there exists a maximal consistent set $\Delta$ such that $\Gamma\subseteq\Delta$ and $\varphi\not\in\Delta$.
  By the assumption, $\Gamma\cup\{\neg\varphi\}$ is consistent set in 
  $\Phi$. Thus, we obtain a maximal consistent set $\Delta$ in $\Phi$ such that $\Gamma\subseteq\Delta$ and $\varphi\not\in\Delta$ by Lemma 2.
\end{proof}

In this paper, we have adopted the global definition. This means that we cannot construct the unique canonical model with the domain that is the whole set of maximal consistent sets. We need to divide the whole set into some groups so that the truth value of $A_{\theta}\varphi$ is the same within the group. At first, we define a base model restricted by a set of formulas.
\begin{df}
  Let $\Phi$ be the closure of a formula. The base model $M^*$ for $\Phi$ is a tuple $\langle W^*,$$ \{(\mathop{\sim_i})^*\}_{i\in\mathcal{G}}, V^*\rangle$, 
  where:
    \begin{itemize}
      \item $W^*$ $\coloneqq\{\Gamma \mid \Gamma 
      \text{ is a maximal consistent set in }\Phi\}$,
      \item $(\Gamma,\Delta)\in (\mathop{\sim_i})^* \textit{\  iff  \ }
      \{\varphi\mid I_i\varphi\in \Gamma\}\subseteq \Delta$,
      \item $V^*(p)\coloneqq \{\Gamma \mid p\in\Gamma\}$.
    \end{itemize}
\end{df}
\noindent
We define a relation $(\mathop{\approx_{\theta}})^*$ on the model $M^*$ for a closure.
\begin{df}
  Let $\Phi$ be the closure of a formula.
  For each $\theta\in\Theta$, a relation $(\mathop{\approx_{\theta}})^*$ on $W^*$ is defined by $(\Gamma,\Delta) \in (\mathop{\approx_{\theta}})^* \textit{\  iff  \ } \{\varphi\mid [\approx]_{\theta}\varphi\in \Gamma\}\subseteq \Delta$.
\end{df}

Now, we are ready to define a divided model for a set of formulas. A divided model is defined as the model generated by a maximal consistent set from the base model for a set of formulas. 
\begin{df}
Let $\Phi$ be the closure of a formula and $\Lambda$ be a maximal consistent set in $\Phi$. The divided model $M^*_{\Lambda}$ by $\Lambda$ for $\Phi$ is a tuple $\langle W^*_{\Lambda}, \{(\mathop{\sim_i})^*_{\Lambda}\}_{i\in\mathcal{G}}, V^*_{\Lambda},\{(\mathscr{A}_{\theta})^*_{\Lambda}\}_{\theta \in \Theta}\rangle$, where:
  \begin{itemize}
    \item $W^*_{\Lambda}$ $\coloneqq\{\Gamma \mid \Gamma 
    \text{ is a maximal consistent set in }\Phi$, and 
    $(\Lambda, \Gamma)\in \bigcup_{i\in \mathcal{G}, \theta\in\Theta} 
    ((\sim_i)^*\circ (\approx_\theta)^*)^+\}$; 
    \item $(\mathop{\sim_i})^*_{\Lambda} \coloneqq (\mathop{\sim_i})^* \cap 
    (W^*_{\Lambda} \times W^*_{\Lambda})$;
    \item $V^*_{\Lambda}(p)\coloneqq V^*(p) \cap W^*_{\Lambda}$;
    \item $(\mathscr{A}_{\theta})^*_{\Lambda} \coloneqq 
    \{p\mid \text{for all }\Gamma\in W^*_{\Lambda}, A_{\theta} p\in \Gamma\}
    $;
  \end{itemize}
\end{df}
 
We prove that every divided model for Closures is an epistemic model with awareness.
\begin{lemma}
  Let $\Lambda$ be a maximal consistent set in $\Phi$. For every $\varphi$, each divided model by $\Lambda$ for the closure of $\varphi$ is an epistemic model with awareness. 
\end{lemma}
\begin{proof}
  We prove that each divided model for the closure of $\varphi$ satisfies the definition of an epistemic model with awareness. For $W^*_{\Lambda}$ and $V^*_{\Lambda}$, the proofs are trivial.
  \begin{itemize}
  \item For $(\mathop{\sim_i})^*_{\Lambda}$, it can be proven in a similar way as \textbf{S5}. For reflexivity, it is sufficient to prove $(w,w)\in (\mathop{\sim_i})^*_{\Lambda}$ for all $w\in W^*_{\Lambda}$. It follows from $\mathrm{K_T}$. For symmetricity, it is sufficient to prove that for all $w,v\in W$, if $(w,v)\in (\mathop{\sim_i})^*_{\Lambda}$, then $(v,w)\in (\mathop{\sim_i})^*_{\Lambda}$. It follows from $\mathrm{T_L}$ and $\mathrm{5_L}$. For transitivity, it is sufficient to prove that for all $w,v,u\in W^*_{\Lambda}$, if $(w,v)\in (\mathop{\sim_i})^*_{\Lambda}$ and $(v,u)\in (\mathop{\sim_i})^*_{\Lambda}$, then $(w,u)\in (\mathop{\sim_i})^*_{\Lambda}$. It follows through $I_i\varphi\to I_i I_i\varphi$ from $\mathrm{K_L,T_L,5_L}$, and $\mathrm{LG}$.
  \item For $(\mathscr{A}_{\theta})^*_{\Lambda}$, 
  it is sufficient to prove that 
  for every $\theta'\in \Theta$, 
  if $\theta' \preceq \theta$, then  
  $(\mathscr{A}_{\theta})^*_{\Lambda} \subseteq 
  (\mathscr{A}_{\theta'})^*_{\Lambda}$, which means that 
  for all $w$, if $A_{\theta} p \in w$, 
  then $A_{\theta'} p \in w$. 
  Suppose that $A_{\theta} p \in w$. 
  It follows that $A_{\theta'} p\in w$ from $\mathrm{AI}$.
\end{itemize} 
\vspace{-5mm}
\end{proof}

Moreover, we also define a relation $(\mathop{\approx_{\theta}})^*_{\Lambda}$ on a divided model $M^*_{\Lambda}$ for a closure $\Phi$ similarly.
\begin{df}
  Let $\Phi$ be the closure of a formula. For each $\theta\in\Theta$, a relation $(\mathop{\approx_{\theta}})^*_{\Lambda}$ on $W^*_{\Lambda}$ is defined as $(\mathop{\approx_{\theta}})^*_{\Lambda} \cap (W^*_{\Lambda} \times W^*_{\Lambda})$.
\end{df}

\begin{lemma}
  A relation $(\mathop{\approx_{\theta}})^*_{\Lambda}$ on each divided model is an indistinguishable relation $\mathop{\approx_{\theta}}$ on an epistemic model with awareness.
\end{lemma}
\begin{proof}
  At first, it is needed to prove that $(\mathop{\approx_{\theta}})^*_{\Lambda}$ is an equivalence relation. This is proven in the same way as $(\mathop{\sim_i})^*_{\Lambda}$ because both relations are defined essentially in the same way and both operators have the axioms corresponding to $K, T,$ and $5$. It is sufficient to prove that for all $(w,v)\in (\mathop{\approx_{\theta}})^*_{\Lambda}$ and every $p\in (\mathscr{A}_{\theta})^*_{\Lambda}$, if $w\in V^*_{\Lambda}(p)$, then $v\in V^*_{\Lambda}(p)$, and vice versa. From left to right, suppose that $w\in V^*_{\Lambda}(p)$ for every $p\in (\mathscr{A}_{\theta})^*_{\Lambda}$. Thus, $A_{\theta} p\in w$ because of the definition of $(\mathscr{A}_{\theta})^*_{\Lambda}$. It follows that $p\in v$ from $\mathrm{AN[\approx]}$. The reverse direction is proven by using that $(\mathop{\approx_{\theta}})^*_{\Lambda}$ is an equivalence relation.
\end{proof}

The next task is to define a $(C_{\theta})_{\Lambda}$-path on a divided model $M^*_{\Lambda}$ for a closure. 
\begin{df}
  Let $\Phi$ be the closure of a formula. A $(C_{\theta})_{\Lambda}$-path from $\Gamma$ is a sequence $\Gamma_0,\dots,\Gamma_n$ of maximal consistent sets in $\Phi$ such that $(\Gamma_k,\Gamma_{k+1}) \in (\mathop{\sim_i})^*_{\Lambda}\circ (\mathop{\approx_{\theta}})^*_{\Lambda}$ for all $k$, where $i$ is the last member of $\theta$, $0\leq k\leq n$, and $\Gamma_0 = \Gamma$. The length of $\Gamma_0,\dots,\Gamma_n$ is $n$. A $\varphi$-path is a sequence $\Gamma_0,\dots,\Gamma_n$ of maximal consistent sets in $\Phi$ such that $\varphi\in\Gamma_k$ for all $k$, where $0\leq k\leq n$.
\end{df}  

\begin{lemma}
  Let $\Phi$ be the closure of a formula; $M^*_{\Lambda}$ be a divided model for $\Phi$; $\Gamma$, $\Delta$ be maximal consistent sets in $W^*_{\Lambda}$; $i$ be the last member of $\theta$.
  If $\bigwedge\Gamma \wedge \neg[\approx]_{\theta} I_i \neg \bigwedge\Delta$ is consistent, then $(\Gamma,\Delta)\in ((\mathop{\sim_i})^*_{\Lambda}\circ (\mathop{\approx_{\theta}})^*_{\Lambda})$ for each divided model $M^*_{\Lambda}$ for $\Phi$. 
\end{lemma}

\begin{proof}
It is sufficient to prove that we suppose that $\bigwedge\Gamma \wedge \neg[\approx]_{\theta} I_i \neg \bigwedge\Delta$ is consistent and $[\approx]_{\theta} I_i\varphi \in \Gamma$ for every $\varphi$, and we obtains $\varphi\in\Delta$.
We prove it by the contradiction. By the assumptions, $[\approx]_{\theta} I_i\varphi \wedge \neg[\approx]_{\theta} I_i\neg \bigwedge\Delta$ is consistent. Suppose that $\varphi\not\in\Delta$, then $\neg\varphi\in\Delta$. It follows that $[\approx]_{\theta} I_i\varphi \wedge\neg[\approx]_{\theta} I_i\varphi$ is consistent. This is the contradiction. Thus, $\varphi\in\Delta$. 
\end{proof}

\begin{lemma}
  Let $\Phi$ be the closure of a formula and $M^*_{\Lambda}$ be a divided model for $\Phi$.
  If $C_{\theta}\varphi\in\Phi$ and $\Gamma \in W^*_{\Lambda}$, then $C_{\theta}\varphi\in\Gamma$ iff every $(C_{\theta})_{\Lambda}$-path from $\Gamma$ is a $\varphi$-path and a $C_{\theta}\varphi$-path. 
\end{lemma}

\begin{proof}
($\Rightarrow$) We prove it by induction on the length of a $(C_{\theta})_{\Lambda}$-path.
\begin{itemize}
\setlength{\itemsep}{-1pt}
  \item For the base case, suppose that the length of a $(C_{\theta})_{\Lambda}$-path is $0$, $C_{\theta} \varphi\in\Phi$, and $C_{\theta}\varphi\in\Gamma$. $\Gamma =\Gamma_0 = \Gamma_n$. By $\mathrm{MIX}$, $\varphi\in\Gamma$.
  \item For induction steps, suppose that the length of a $(C_{\theta})_{\Lambda}$-path is $k+1$, $C_{\theta} \varphi\in\Phi$, and $C_{\theta}\varphi\in\Gamma$. By the induction hypothesis, $C_{\theta}\varphi\in\Gamma_k$. From $\mathrm{MIX}$ and the definition of $(\mathop{\sim_i})^*_{\Lambda}$ and $(\mathop{\approx_{\theta}})^*_{\Lambda}$, it follows that $\varphi$ and $C_{\theta}\varphi\in \Gamma_{k+1}$. 
\end{itemize}

($\Leftarrow$) Let $S((C_{\theta})_{\Lambda},\varphi)$ be a set of 
maximal consistent sets $\Delta$ in $W^*_{\Lambda}$ such that 
every $(C_{\theta})_{\Lambda}$-path from $\Delta$ is a $\varphi$-path. 
We introduce a special formula: 
\[\chi = \bigvee_{\Delta\in S((C_{\theta})_{\Lambda},\varphi)}\bigwedge \Delta.\] 

Suppose that every $(C_{\theta})_{\Lambda}$-path from $\Gamma$ is a $\varphi$-path. First, we need to prove these three: 
\begin{align*}
(1)\ \vdash &\bigwedge\Gamma\to \chi; 
\quad (2)\ \vdash \chi\to\varphi;
\quad (3)\ \vdash \chi\to [\approx]_{\theta} I_i\chi, 
\end{align*}
where $i$ is the last member of $\theta$.
\begin{itemize}
\setlength{\itemsep}{-1pt}
\item For (1), $\Gamma\in S((C_{\theta})_{\Lambda},\varphi)$ by the assumption. 
Thus, $\vdash \bigwedge\Gamma\to \chi$. 
\item For (2), since every $(C_{\theta})_{\Lambda}$-path from $\Delta$ is a $\varphi$-path, $\varphi\in\Delta$ for every $\Delta \in S((C_{\theta})_{\Lambda},\varphi)$. Thus, $\varphi$ is derived from $\chi$. 
\item For (3), we prove it by contradiction. Suppose $\chi\wedge\neg[\approx]_{\theta} I_i\chi$ is consistent. By the construction of $\chi$, there exists $\Delta$ such that $\bigwedge\Delta\wedge\neg[\approx]_{\theta} I_i\chi$ is consistent. $W^*_{\Lambda}$ can be interpreted as the whole set of combinations of formulas in $\Phi$ that satisfies the condition, and $\chi$ can be interpreted as a particular set of combinations. The formula $\neg\bigvee_{\Lambda\in W^*_{\Lambda}\setminus S(C_{\theta},\varphi)} \bigwedge\Lambda$ is equivalent to $\chi$ because the negation of the other combinations can represent a particular set of combinations represented by $\chi$. Therefore, $\bigwedge\Delta\wedge\neg[\approx]_{\theta} I_i \neg\bigvee_{\Xi\in W^*_{\Lambda}\setminus S(C_{\theta},\varphi)} \bigwedge\Xi$ is consistent. There is $\Xi$ such that $\bigwedge\Delta\wedge\neg[\approx]_{\theta} I_i \neg\bigwedge\Xi$ is consistent. By Lemma 6, $(\Delta,\Xi)\in (\mathop{\sim_i})^*_{\Lambda}\circ(\mathop{\approx_{\theta}})^*_{\Lambda}$ for every divided model $M^*_{\Lambda}$. There exists a $(C_{\theta})_{\Lambda}$-path from $\Delta$ that is not a $\varphi$-path. This is a contradiction. Thus, $\vdash \chi\to [\approx]_{\theta} I_i\chi$.
\end{itemize}
By $(3)$ and $\mathrm{CG}$, $\vdash C_{\theta}(\chi\to [\approx]_{\theta} 
I_i\chi)$. It follows that $\vdash\chi\to C_{\theta} \chi$ from $\mathrm{IND}$. 
By $(1)$ and $(2)$, $\vdash\bigwedge \Gamma\to C_{\theta}\varphi$.
Thus, $C_{\theta}\varphi\in\Gamma$.    
\end{proof}


The next task is the proof of the truth lemma, which is that for every divided model, a true formula at a world in a model is included in the world.
\begin{lemma}
  Let $\Phi$ be the closure of a formula and $M^*_{\Lambda}$ be a divided model for $\Phi$. For all $w\in W^*_{\Lambda}$ and every $\varphi\in \Phi$, $M^*_{\Lambda},w\vDash\varphi$ iff $\varphi\in w$.
\end{lemma}
\begin{proof}
  We prove it by induction on the structure of formulas. 
  Proofs of the base case and logical connectives are straightforward. 
\begin{itemize}
  \item For the case of $A_{\theta} \varphi$, 
  we prove it by induction 
  on the structure of $\varphi$.
  \begin{itemize}
  \item For the base case, suppose that $M^*_{\Lambda},w\vDash A_{\theta} p$, then $p\in\{q\mid \text{for all } w, A_{\theta} q\in w\}$. Thus, $A_{\theta} p\in w$. As for the reverse direction, suppose that $A_{\theta} p \in w$. From $\mathrm{KA, NKA}$, $\bigwedge_{\theta'\in\Theta} C_{\theta'} A_{\theta} p\in w$. It follows that for all $v$ in $W^*_{\Lambda}$, $A_{\theta} p\in v$ from $\mathrm{MIX}$. Thus, $M^*_{\Lambda},w\vDash A_{\theta} p$.
  \item For the other cases, we obtain the desired proof by induction hypothesis and decomposing the formula using corresponding axioms: $\mathrm{AN}$, $\mathrm{ACN}$, $\mathrm{AA}$, $\mathrm{AL}$, $\mathrm{A[\approx]}$, $\mathrm{ACM}$, and $\mathrm{AK}$. We show the strategy by taking $\mathrm{AN}$ as an example.
    \begin{itemize}
      \item For the case of $\neg\psi$, suppose that $M^*_{\Lambda},w\vDash A_{\theta}\neg\psi$, which means that $At(\neg \psi)\subseteq \{q\mid \text{for all } w, A_{\theta} q\in w\}$. Since $At(\neg \psi) = At(\psi)$, $M^*_{\Lambda},w\vDash A_{\theta}\psi$. $A_{\theta}\neg \psi\in w$ by induction hypothesis and $\mathrm{AN}$. The reverse direction is proven similarly.
    \end{itemize}
  \end{itemize}
  \item For the case of $I_i \varphi$ and $[\approx]_{\theta}\varphi$, it is proven in a similar strategy as \textbf{S5}. We show only the case of $I_i\varphi$ here. From left to right, suppose that $M^*_{\Lambda},w\vDash I_i\varphi$, which means that for all $v$ such that $\{\psi\mid I_i\psi\in w\}\subseteq v$, $M^*_{\Lambda},v\vDash \varphi$. By induction hypothesis, it follows that for all $v$ such that $\{\psi\mid I_i\psi\in w\}\subseteq v$, $\varphi\in v$.
    By Lemma 3, $\{\psi\mid I_i\psi\in w\}\vdash\varphi$. There exists a finite subset $\{\psi_1,\dots,\psi_n\}$ of $\{\psi\mid I_i\psi\in w\}$ such that $\vdash (\psi_1\wedge\dots\wedge \psi_n)\to \varphi$. It follows that $\vdash (I_i\psi_1\wedge\dots\wedge I_i\psi_n) \to I_i\varphi$ from $\mathrm{LG}$ and $\mathrm{K_L}$. Since a set $\{I_i\psi_1,\dots,I_i\psi_n\}$ is included in $w$, $w\vdash I_i\varphi$. Thus, $I_i\varphi\in w$. The reverse direction is straightforward.
  \item For the case of $C_{\theta}$, suppose that $M^*_{\Lambda},w\vDash C_{\theta}\varphi$. For all $v$ such that $(w,v)\in ( (\mathop{\sim_i})^*_{\Lambda}\circ (\mathop{\approx_{\theta}})^*_{\Lambda})^+$, $M^*_{\Lambda},v\vDash \varphi$. It means every $(C_{\theta})_{\Lambda}$-path from $w$ is a $\varphi$-path. By Lemma 7, $C_{\theta}\varphi\in w$.
  \item For the case of $E_{\theta}$, it is proven by decomposing into $A_{\theta}$ and $E_{\theta}$. Suppose that $M^*_{\Lambda},w\vDash E_{\theta}\varphi$, then $M^*_{\Lambda},w\vDash A_{\theta}\varphi$ and $M^*_{\Lambda},w\vDash C_{\theta}\varphi$. It follows from the proofs for $A_{\theta}$ and $C_{\theta}$ that $A_{\theta}\varphi\in w$ and $C_{\theta}\varphi\in w$. By $\mathrm{KAC}$, $E_{\theta}\varphi\in w$
  \end{itemize}
  \end{proof}

\begin{lemma}
Let $\Phi$ be the closure of a formula and $\Gamma$ be a maximal consistent set in $\Phi$. For every $\varphi\in\Phi$, if $\varphi\in \Gamma$ for every maximal consistent set $\Gamma$ in $\Phi$, then $\vdash\varphi$.
\end{lemma}
\begin{proof}
  We prove it by contraposition. Suppose $\nvdash\varphi$, then $\{\neg \varphi\}$ is a consistent set in $\Phi$. By Lemma 2, there is a maximal consistent set in $\Phi$ that does not contain $\varphi$.
\end{proof}

Finally, we prove the completeness theorem. To prove that a formula is a theorem, it is necessary that every maximal consistent set in the closure of the formula includes the formula. Since the assumption of the theorem leads to the formula being true at all worlds in every divided model for the closure of the formula, we obtain the desired fact.
\begin{thm}
  For every $\varphi\in\mathcal{L}_{(\mathcal{P,G},\Theta)}$, if $\vDash\varphi$, then $\vdash\varphi$.
\end{thm}
\begin{proof}
  Suppose that $\vDash\varphi$. For every divided models $M^*_{\Lambda}$ for the closure of $\varphi$, $M^*_{\Lambda},w\vDash\varphi$ by Lemma 4.  Thus, for every maximal consistent set $w$ in $\Phi$, $\varphi\in w$ by Lemma 8. Therefore, $\vdash\varphi$ by Lemma 9.
\end{proof}

\section{Related Work}
There are several studies that are based on a similar idea as this paper, which is to connect a state of an agent's awareness with a distinction of possible worlds. \cite{van2009awareness,ditmarsch2011becoming,van2018implicit} provided \textit{awareness bisimilation}. This depends on a set of atomic propositions and is defined between the models that cannot be distinguished by formulas consisting only of the atomic propositions. Their logic represents the indistinguishment of possible worlds by lack of awareness, as with our logic. The main difference is the definition of explicit knowledge. On the definition of explicit knowledge in our logic, an agent considers from an evaluate world $w$ not only to a world $s$ with a different valuation only for propositions of which an agent is not aware but also to a world $t$ accessible from $s$, unlike their logic.
\cite{yudai2022-1} defined an indistinguishable relation and proposed a different definition of explicit knowledge from \cite{fagin1988belief}. Those ideas are used as the basis of our logic. They also incorporated the notion of an agent's viewpoint to represent what ``an agent is aware of a proposition from another agent's viewpoint,'' which is represented as $A^i_j \varphi$. Our logic proposes the notion of $a$'s belief in $b$'s awareness instead of the notion of agent $b$'s awareness from $a$'s viewpoint. Moreover, their logic does not handle a chain of viewpoints with a depth of more than 3, but our logic handles a chain of belief for awareness with a finite depth.

Purely semantic approaches to awareness are active in the field of economics. It is also called the event-based approach, in which the notion of \textit{events} that are represented as a set of possible worlds is employied. An agent's awareness and knowledge are expressed as an operator on events. The logic system proposed in \cite{modica1994awareness, modica1999unawareness} is the early work of the approach. \cite{halpern2001alternative} found it to be equivalent to a part of the logic in \cite{fagin1988belief}. \cite{heifetz2006interactive,heifetza2008canonical} proposed \textit{unawareness frame} that is a lattice on disjoint state-spaces with a partial order. One disjoint state-space captures one particular horizon of propositions and works a restriction of propositions instead of an awareness set. A state-space $S$ is greater than $S'$ means that states of $S$ describe situations with a richer vocabulary than states of $S'$. One state-space can be interpreted as subjective descriptions of situations in an agent's mind. The formalization is applied to a generalization of \cite{aumann1976agreeing} \cite{heifetz2013unawareness}. Although this study does not explicitly have a chain concerning awareness, we may find some relation to our logic. 

In game theory, a formalization of the concept of awareness/unawareness has been studied in order to incorporate it into games. \cite{feinberg2005games,Feinberg+2021+433+488} defined a game with incomplete awareness and also defined a higher order awareness or views as a sequence of agents, which are an iteration of agents' awareness or views. 

\textit{Epistemic worlds semantics} provided in \cite{kaneko2003epistemic} is for epistemic logic of shallow depths in which nested occurrences of belief operators are bounded. The domain in the semantics is the disjoint union of a set $W_e$ of possible worlds indexed by a sequence of agents, and a nested belief is evaluated within the set of possible worlds indexed by the sequence of agents corresponding to the order of agents in the nested belief. Although thier study does not treat the concept of awareness, the idea that the subject in a nested belief or knowledge bounds them is similar as our logic. The logic in \cite{grossi2015ceteris} has an operator $[X]$ representing that a proposition holds at all the possible worlds with the same valuation for all the elements of a given set $X$ of atomic propositions. \textit{Team semantics} \cite{sano2015axiomatizing} used in dependence logic formalizes that an atomic proposition is true at all the elements of a given subset of possible worlds. This subset called a \textit{term}, supports a formula.


\section{Conclusion}
This paper provided Awareness Logic with Partitions and Chains ($\mathcal{ALPC}$) for formalizing nested explicit knowledge. We incorporated an indistinguish relation on possible worlds that changes in conjunction with a state of awareness and a chain of belief for awareness. These reflect two ideas: the effect of an agent's awareness on her own reasoning and a state of an agent's awareness in another agent's mind. Employing this framework, we demonstrated an example where each agent's knowledge that each other agent who has different states of awareness actually knows at a given moment, which is nested explicit knowledge, could be logically represented.

Our contributions of this logic are two-fold. From a logical viewpoint, 
we defined the syntax and the semantics of $\mathcal{ALPC}$ and proved its completeness. From the viewpoint of applicability to the real world, we showed the architecture to explain the strategic behavior of a rational agent in society or game theory. We expect that the logic can be applied to computer science and game theory by describing and analyzing strategic behavior in a game and practical agent communication.

There are several directions in the future. Extending our logic to represent knowledge of a group is one direction. In particular, formalizing common knowledge is an interesting topic. In decision-making in strategic situations where a player makes a decision with knowledge of the other's action, it is valuable that a proposition is common knowledge from the perspective of an outside observer with complete awareness. However, it is equally valuable that a proposition is common knowledge from the perspective of an agent with incomplete awareness. Our logic can be used to represent such common knowledge. Furthermore, incorporating an epistemic action that changes a state of awareness is another direction. This extension makes it possible to formalize changing knowledge by communicative actions \cite{van2011logical}. There are already several studies on dynamics of awareness \cite{fernandez2021awareness,grossi2015syntactic,van2010dynamics}, and we can refer to those results. On the technical side, exploring a relationship to other related logic mentioned in Section 5 is one of the directions. In addition, our logic is applicable to the studies dealing with multiple agents' reasoning, such as a description and analysis of a game that takes into account players' awareness of possible strategies. 

\section*{Acknowledgements}
I thank Nobu-Yuki Suzuki, Satoshi Tojo, and Kosuke Udatsu for their insightful comments.
\bibliographystyle{abbrv}
\bibliography{refs}

\end{document}